\documentclass[lettersize,onecolumn]{IEEEtran}
\usepackage{amsmath,amsfonts,amsthm,amssymb}
\usepackage{algorithm}
\usepackage{hyperref}
\usepackage{algpseudocode}
\usepackage{array}
\usepackage[caption=false,font=normalsize,labelfont=sf,textfont=sf]{subfig}
\usepackage{textcomp}
\usepackage{url}
\usepackage{verbatim}
\usepackage{graphicx}
\usepackage{cite}
\usepackage{wrapfig}
\usepackage{dsfont}
\usepackage{tikz}
\usepackage{booktabs}
\usetikzlibrary{arrows.meta,calc,positioning}
\newtheorem{theorem}{Theorem}[section]

\newtheorem{lemma}[theorem]{Lemma}

\newif\ifCOM
\COMfalse

\newcommand\eqdef{\overset{\mbox{\tiny def}}{=}}

\begin{document}

\title{New lower bounds for kissing numbers in dimensions $25$--$31$}

\author{Rustem~Takhanov and Stanislav~Yun
\thanks{R.~Takhanov and S.~Yun are with the Mathematics Department, Nazarbayev University, and Nazarbayev University Research Administration, Astana, Kazakhstan e-mail: rustem.takhanov@nu.edu.kz.}
}



\maketitle
\begin{abstract}
The kissing number in dimension $d$ is the largest number of non-overlapping congruent spheres that can simultaneously touch a central sphere of the same size. We study dimensions $25$-$31$, where the best previous constructions are based on Leech lifting from the optimal kissing configuration in dimension $24$. Our method exploits the absence of contacts between the unlifted bulk and the block consisting of lifted and auxiliary vectors. Rotating this block while keeping the bulk fixed creates room for two antipodal points in dimensions $26$, $27$, and $28$, and one point in dimension $29$. 
 Three further modifications yield improvements in dimensions $25$, $30$ and $31$: (a) a nonorthogonal diagonal linear deformation of the lifted block admits two antipodal points in dimension $25$; (b) rotating the additional coordinates of the lifted vectors and then applying a small orthogonal transformation to the resulting lifted block as a whole admits two antipodal points in dimension $30$; (c) rotating only the additional coordinates of the lifted vectors admits four nonantipodal points in dimension $31$. Together, these constructions yield the new lower bounds
$\tau_{25}\geq 197058$, $\tau_{26}\geq 198552$,
$\tau_{27}\geq 200046$, $\tau_{28}\geq 204522$,
$\tau_{29}\geq 209497$, $\tau_{30}\ge 220442$, and $\tau_{31}\geq 238354$.
\end{abstract}

\begin{IEEEkeywords}
Kissing number, Leech lifting, Leech lattice, optimization.
\end{IEEEkeywords}

\section{Introduction}
\label{sec:introduction}

A kissing configuration in dimension $d$ is a finite set $C\subset {\mathbb S}^{d-1}$ such that $u^\top v\leq \frac12$ for every pair of distinct points $u,v\in C$.
The kissing number $\tau_d$ is the largest possible cardinality of such a
configuration. Equivalently, it is the largest number of non-overlapping
unit spheres that can simultaneously touch a central unit sphere.
Every explicit kissing configuration therefore gives a lower bound
$\tau_d\geq |C|$.

In dimensions $d\leq 8$, the best known lower bounds are attained by
configurations already known in the nineteenth century or earlier;
the classical root systems provide examples of all these
sizes~\cite{Schlafli1901,KorkineZolotareff1873}.
Their persistence does not mean that the corresponding optimality
questions have all been settled: the exact kissing numbers in
dimensions $5$, $6$, and $7$ remain unknown.
Moreover, the configurations attaining these cardinalities need not
be unique. Geometrically distinct examples have continued to emerge,
including the recent five-dimensional constructions of
Sz\"oll\H{o}si~\cite{Szollosi2023} and Cohn--Rajagopal~\cite{CohnRajagopal2026}.

Beyond dimension $8$, improvements have come from a wider range of
constructions. The classical work of Leech and
Sloane~\cite{LeechSloane1971} connected kissing configurations with
error-correcting codes and established several long-standing records.
Zinoviev and Ericson~\cite{ZinovievEricson1999} subsequently obtained
the bound $\tau_{13}\geq 1154$.
More recently, Ganzhinov~\cite{Ganzhinov2025} improved the lower bounds
in dimensions $10$, $11$, and $14$, while Takhanov, Assylbekov, and
Yun~\cite{TakhanovAssylbekovYun2026} constructed an $841$-point
configuration in dimension $12$.  
The latter configuration can be understood as a deformation of an $840$-point kissing arrangement admitting an explicit quaternionic description, with one additional point accommodated by the deformation~\cite{takhanov2026quaternionic}. Earlier, Takhanov and Yun~\cite{takhanov2026classification} classified the maximum independent sets in the signed Johnson graph $J_{\pm}(12,4)$, obtaining $1579$ pairwise non-isometric $840$-point kissing arrangements.
Dimension $11$ has also proved particularly fruitful for searches
using artificial intelligence: AlphaEvolve improved Ganzhinov's
$592$-point bound to $593$~\cite{AlphaEvolve2025}, and collaborative
agents on EinsteinArena subsequently reached
$604$~\cite{EinsteinArena2026}.

At dimension $16$, the $4320$ minimal vectors of the Barnes--Wall
lattice remain unsurpassed~\cite{BarnesWall1959}.
There is nevertheless a distinct \emph{odd} $16$-dimensional kissing
configuration of the same size.
This alternative is useful in higher dimensions: Cohn and
Li~\cite{CohnLi2024} improved the lower bounds in every dimension
from $17$ through $21$.
Ho~\cite{Ho2026} then improved their dimension-$19$ bound to $11948$.
The records in dimensions $22$ and $23$ still come from Leech's
constructions~\cite{Leech1967}. In dimension $24$, the normalized minimal vectors of
the Leech lattice form a kissing configuration of size $196560$,
which attains the exact kissing
number~\cite{Leech1967,Levenshtein1979,OdlyzkoSloane1979}.

This paper concerns dimensions $d=24+k$, where $1\leq k\leq 7$.
The strongest known constructions in this range remain closely tied
to the Leech lattice. Starting from its normalized minimal shell
in $\mathbb R^{24}\oplus\{0\}$, one replaces selected vectors by
lifted vectors in $\mathbb R^{24}\oplus\mathbb R^k$ (one-to-many) and may add an
auxiliary kissing configuration supported entirely in the second
factor. This lifting construction was introduced by Cohn, Jiao,
Kumar, and Torquato~\cite{CJKT2011}.
It allows considerable choice in the Leech subsets selected for
lifting, the directions used to lift them, and the auxiliary
configuration, subject to the required inner-product conditions.
Kallal, Kan, and Wang~\cite{KallalKanWang2017} improved the bounds in
dimensions $25$--$31$ by finding larger suitable Leech subsets and
improving the selection of disjoint subsets.
More recently, the reinforcement learning system PackingStar
produced the records that serve as our starting points, using
$496$-point Leech subcodes and refined decompositions of the
lower-dimensional configurations~\cite{PackingStar}.

Our approach exploits a simple geometric feature of these
constructions. In this construction, the set of vectors can be decomposed as
\[
  C_d=B_d\sqcup M_d,
\]
where $B_d$ is the unlifted bulk and $M_d$ consists of all lifted
vectors together with the auxiliary configuration.
There are no contacts between the two blocks: every inner product
between a bulk vector and a lifted vector is at most
$1/\sqrt6<1/2$, while the auxiliary vectors are orthogonal to the
bulk. Consequently, the whole block $M_d$ can undergo sufficiently
small rotations relative to $B_d$ without violating the kissing
condition. Such rotations preserve all inner products within
$M_d$, while changing the space available for additional points.
We use this freedom to search jointly for a rotation
$Q\in\operatorname{SO}(d)$ and a new unit vector $x$, through the
constrained minimax problem
\begin{equation}
  \min_{\substack{x\in{\mathbb S}^{d-1},\ Q\in\operatorname{SO}(d)\\
                  b^\top Qm\leq 1/2\quad(b\in B_d,\ m\in M_d)}}
       \ \max_{y\in B_d\cup QM_d} x^\top y.
  \label{eq:introduction-minimax}
\end{equation}
Any feasible pair with objective value at most $1/2$ supplies an
additional point. We then check whether its antipode can also be
included.

Applying this procedure to the PackingStar configurations yields
the new lower bounds
\begin{equation}
  \tau_{26}\geq 198552,\qquad
  \tau_{27}\geq 200046,\qquad
  \tau_{28}\geq 204522,\qquad
  \tau_{29}\geq 209497.
  \label{eq:introduction-new-bounds}
\end{equation}
The improvements consist of two antipodal points in each of
dimensions $26$, $27$, and $28$, and one point in dimension $29$.
The resulting configurations are verified with explicit error
bounds accounting for the numerical representation of their
coordinates and rotations.

Three other modifications of Leech lifting yield
\begin{equation}
  \tau_{25}\geq 197058,\qquad\tau_{30}\geq 220442,\qquad
  \tau_{31}\geq 238354.
  \label{eq:introduction-additional-bounds}
\end{equation}
In dimension $25$, we apply a nonorthogonal diagonal linear
transformation to the lifted block, moving its vectors closer
to the ``equator''
$({\mathbb R}^{24}\times\{0\})\cap{\mathbb S}^{24}$
while preserving their unit norms. This allows us to add the two poles
$\pm e_{25}$ to the arrangement. The validity of this construction
follows from direct inner-product calculations.

In dimension $31$, we apply a common rotation in
$\operatorname{SO}(7)$ to the additional coordinates of all
lifted vectors, keeping both the bulk and the auxiliary block
fixed. This permits the addition of four nonantipodal points. In dimension $30$, we additionally apply a small orthogonal transformation to the resulting lifted block as a whole, allowing the addition of two antipodal points. Constructions are certified
using exact arithmetic and rigorous interval bounds that
account for normalization and numerical representations
of rotations.

\section{Leech lifting and rigid block rotations}
\label{sec:leech-lifting}

We are interested in the case $d=24+k$, where $1\leq k\leq 7$.
Let $\mathcal L_{24}\subset {\mathbb S}^{23}$ denote the normalized minimal vectors of the Leech lattice, of cardinality $196560$.

We use Leech lifting construction from~\cite{CJKT2011}, improved in~\cite{KallalKanWang2017}.  
\begin{lemma}[Leech lifting\cite{CJKT2011,KallalKanWang2017}]\label{leech-lifting} Suppose that
$S_1,\ldots,S_r\subset\mathcal L_{24}$ are pairwise disjoint and satisfy
\begin{equation}
  u^\top v\leq \frac14
  \qquad(u,v\in S_i,\ u\ne v).
  \label{eq:leech-subcodes}
\end{equation}
Let $K_k=T_1\sqcup\cdots\sqcup T_r\subset {\mathbb S}^{k-1}$ be a kissing
configuration whose parts satisfy
\begin{equation}
  t^\top s\leq-\frac12
  \qquad(t,s\in T_i,\ t\ne s).
  \label{eq:auxiliary-parts}
\end{equation}
Let $K'_k\subset {\mathbb S}^{k-1}$ be a kissing configuration satisfying
$w^\top t\leq\frac{\sqrt3}{2}$ for all $w\in K'_k$ and $t\in K_k$. In $\mathbb R^d=\mathbb R^{24}\oplus\mathbb R^k$, define the bulk,
the lifted block, and the auxiliary block by
\begin{align}
  B_d
    &=\left\{(u,0)\mid u\in\mathcal L_{24}
                     \setminus\bigcup_{i=1}^r S_i\right\},
       \label{eq:bulk}\\
   L_d
    &=\bigsqcup_{i=1}^r
      \left\{\left(\sqrt{\frac23}\,u,
                         \sqrt{\frac13}\,t\right)\mid
                  u\in S_i,\ t\in T_i\right\},
       \label{eq:lifted-block}\\
F_d
&=\{(0,w)\mid w\in K'_k\}.
  \label{eq:frame-block}
\end{align}
Then, $B_d\sqcup L_d\sqcup F_d$ is a kissing configuration of size
\begin{equation}
  196560+\sum_{i=1}^r (|T_i|-1)|S_i|+|K'_k|.
  \label{eq:leech-size}
\end{equation}
\end{lemma}
PackingStar uses $496$-point sets $S_i$~\cite{PackingStar}.  The corresponding values of
$\sum_i(|T_i|-1)$ for $d=25,\ldots,31$ are
$1,4,7,16,26,48,84$, and the auxiliary block sizes are
$0,6,12,24,40,72,126$, respectively.

Note that the bulk and the system "lifted and auxiliary blocks" do not have contacts. Indeed, a bulk point and a lifted point have inner product at most
\begin{equation}
  \sqrt{\frac23}\,\frac12=\frac1{\sqrt6}<\frac12,
  \label{eq:initial-cross-slack}
\end{equation}
and the auxiliary block is orthogonal to the bulk. This is the geometric resource that our kissing arrangement construction is based on.

We keep the bulk $B_d$ fixed and rotate the entire block $M_d \eqdef L_d\sqcup F_d$, by one matrix $Q\in\operatorname{SO}(d)$, i.e.
\begin{equation}
  C_d(Q)\eqdef B_d\cup Q M_d,
  \qquad QM_d\eqdef \{Q\ell\mid \ell\in M_d\}.
  \label{eq:rotated-configuration}
\end{equation}
For nonempty finite sets $U,V\subset {\mathbb S}^{d-1}$, let us denote
\begin{equation}
  m(U,V)\eqdef\max_{u\in U,\ v\in V}u^\top v.
  \label{eq:maxdot-definition}
\end{equation}
Rotation preserves all inner products within $M_d$, so $C_d(Q)$ is
kissing precisely when
\begin{equation}
  m(B_d,QM_d)\leq\frac12.
  \label{eq:cross-feasibility}
\end{equation}
The strict initial separation in~\eqref{eq:initial-cross-slack} leaves
room for relative rotations: since all points have unit norm,
\begin{equation}
  m(B_d,Q M_d)\leq\frac1{\sqrt6}+\|Q-I\|_2.
  \label{eq:rotation-neighborhood}
\end{equation}
In particular, every $Q\in\operatorname{SO}(d)$ with
$\|Q-I\|_2\leq \frac12-\frac{1}{\sqrt6}$ is feasible.

\begin{figure}[t]
  \centering
  \includegraphics[width=\textwidth]{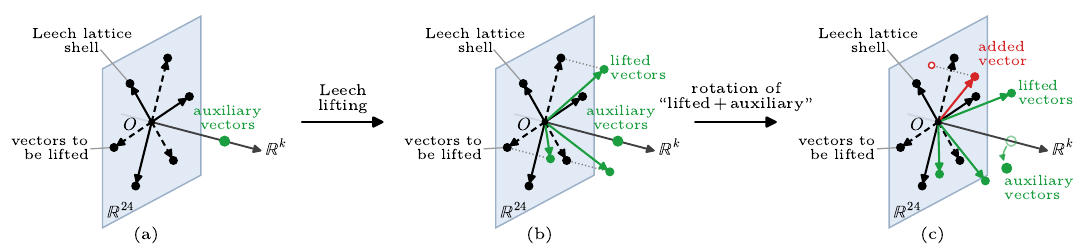}
  \caption{Three-dimensional schematic of Leech lifting and rotation. The vertical plane represents $\mathbb{R}^{24}$, while the horizontal axis, perpendicular to the plane, represents the additional factor $\mathbb{R}^{k}$. The green point initially on the axis represents the auxiliary kissing configuration. The six black vectors pointing to the vertices of a hexagon represent the Leech lattice shell, with the three dashed vectors indicating those selected for lifting. The three green vectors represent the corresponding lifted vectors. In the final panel, the lifted vectors and the auxiliary configuration are rotated together, opening a gap for the new vector shown in red.}
  \label{fig:lifting}
\end{figure}

For a feasible $Q$, a \emph{hole} is a point $x\in {\mathbb S}^{d-1}$ satisfying
$x^\top y\leq \frac12$ for every $y\in C_d(Q)$.  Such a point can be added
to the configuration.  Searching for such a hole is equivalent to solving the optimization task
\begin{equation}
  \begin{aligned}
    \mu_d^*\eqdef {}&\min_{\substack{x\in {\mathbb S}^{d-1},\ Q\in\operatorname{SO}(d)\\
                         m(B_d,QM_d)\leq1/2}} f_d(x,Q),\\[-2pt]
  \end{aligned}
  \label{eq:joint-minimax}
\end{equation}
where $f_d(x,Q)=m(B_d\sqcup QM_d,\{x\})$.
Any feasible pair with $f_d(x,Q)\leq \frac12$ gives an extension, regardless
of whether it attains this minimum.  The strict inequality provides a
margin for numerical verification. The idea of the joint rotation of lifted and auxiliary vectors is schematically presented on Figure~\ref{fig:lifting}.

\section{Experimental results}
\label{sec:experimental-results}

We use a sequential linear programming method with a trust region
to search for a feasible pair $(x,Q)$ with small $f_d(x,Q)$.
After a first hole is found, we freeze its rotation $Q$ and set
$C=C_d(Q)$. If possible, we add the antipode of $x$ to $C\cup \{x\}$. We check that no new points can be added to the latter configuration.

We evaluated the saved runs in dimensions $25$ through $31$ using the PackingStar constructions.
Table~\ref{tab:new-antipodal-results} gives the original size,
the verified size after extension, and the three inter-block maxima.
Here $A_d$ denotes the \emph{entire} selected set of added points.
Its size was $2,2,2,1$
for dimensions $26,27,28,29$, respectively. 

\begin{table}[htbp]
\centering
\small
\setlength{\tabcolsep}{4pt}
\begin{tabular}{rrrrrr}
\toprule
$d$ & Old size & New size & $m(B_d,Q_dM_d)$
    & $m(B_d,A_d)$ & $m(A_d,Q_dM_d)$ \\
\midrule
25 & 197056 & 197056 & 0.499999987678 & --- & --- \\
26 & 198550 & {\bf 198552} & 0.499999983368 & 0.485178375953 & 0.485178375957 \\
27 & 200044 & {\bf 200046} & 0.499999979377 & 0.458393985560 & 0.458457922587 \\
28 & 204520 & {\bf 204522} & 0.499999960555 & 0.461664752860 & 0.463443808987 \\
29 & 209496 & {\bf 209497} & 0.499999718198 & 0.491349037725 & 0.491349039293 \\
30 & 220440 & 220440 & 0.452940397167 & --- & --- \\
31 & 238350 & 238350 & 0.465503357966 & --- & --- \\
\bottomrule
\end{tabular}
\caption{Verified sizes and inter-block maxima.
The displayed maxima are numerical evaluations rounded to twelve decimal places.}
\label{tab:new-antipodal-results}
\begin{minipage}{0.97\textwidth}
\footnotesize
For $d=26,27,28$, $A_d=\{x_d,-x_d\}$; for $d=29$, $A_d=\{x_d\}$.
For $d=25,30,31$, $A_d=\varnothing$, so no point was added and the last
two maxima are undefined. 
\end{minipage}
\end{table}

For $d=26,27,28$, the selected set is an antipodal pair
$A_d=\{x_d,-x_d\}$, whose mutual inner product is exactly $-1$.
For $d=29$, it is the singleton $A_d=\{x_d\}$.
Unsuccessful searches for $d=25,30,31$ do not establish
nonexistence of holes in this optimization-based approach. For $d=25$, in the following section we 
describe a different way to add an antipodal pair, based on finding a suitable affine transformation $Q$.

The largest certified upper bound on a
bulk vs rotated-block maximum across these runs was less than
$0.499999983368$. Consequently the maximum
inner product between distinct points of each verified full
configuration is exactly $\frac12$.

Writing $\tau_d$ for the kissing number in dimension $d$, the four
verified extensions establish
\begin{equation}
  \tau_{26}\geq198552,\qquad
  \tau_{27}\geq200046,\qquad
  \tau_{28}\geq204522,\qquad
  \tau_{29}\geq209497.
  \label{eq:verified-lower-bounds}
\end{equation}
These improve the respective PackingStar configurations by
$2,2,2,1$ points.  The experiments do not certify global optimality of
the rotations or exhaust all possible extensions. The coordinate files of the new kissing configurations, together with the outputs of our numerical searches, are available in the accompanying GitHub repository at \url{https://github.com/k-nic/Leech_lifting}.

\section{The case of $d=25$}
\label{sec:specifics-d25}

In dimension $25$, a modification of the lifting coefficients from Lemma~\ref{leech-lifting} allows us
to add the two poles of the additional one-dimensional factor.
Take $k=1$, $r=1$, and $T_1=\{-1,1\}\subset\mathbb S^0$, and let
$S_1\subset\mathcal L_{24}$ be the $496$-point subset supplied
by PackingStar~\cite{PackingStar}, satisfying
\eqref{eq:leech-subcodes}.
Since $K_1=\mathbb S^0$, the condition on $K'_1$ forces
$K'_1=\varnothing$, and hence $F_{25}=\varnothing$.
The construction above therefore gives $196560+496=197056$ points.

Keep the bulk $B_{25}$ defined in \eqref{eq:bulk}, and introduce a
variable lifting height $0<h<1$ by setting
\begin{equation}
  \begin{aligned}
  \ell_\varepsilon(u, h)
    &=\bigl(\sqrt{1-h^2}\,u,\varepsilon h\bigr),\quad \varepsilon\in\{-1,1\},\\
  L_{25}(h)
    &=\{\ell_\varepsilon(u, h)\mid
        u\in S_1,\ \varepsilon\in\{-1,1\}\}.
  \end{aligned}
  \label{eq:d25-variable-lift}
\end{equation}
Thus $L_{25}(\frac{1}{\sqrt3})=L_{25}$ is the original lifted block from Lemma~\ref{leech-lifting}.
All these vectors have norm one. For distinct $u,v\in S_1$ and
arbitrary signs $\varepsilon,\delta\in\{-1,1\}$,
\begin{equation}
  \ell_\varepsilon(u;h)^\top\ell_\delta(v;h)
    =(1-h^2)u^\top v+\varepsilon\delta h^2
    \leq\frac14+\frac34h^2,
  \label{eq:d25-distinct-lifts}
\end{equation}
whereas the two lifts of the same vector satisfy
\begin{equation}
  \ell_+(u;h)^\top\ell_-(u;h)=1-2h^2.
  \label{eq:d25-paired-lifts}
\end{equation}
Moreover, if $z\in\mathcal L_{24}\setminus S_1$ and $u\in S_1$, then
\[
  (z,0)^\top\ell_\varepsilon(u;h)
    =\sqrt{1-h^2}\,z^\top u
    \leq\frac{\sqrt{1-h^2}}2<\frac12.
\]
Pairs within $B_{25}$ already satisfy the kissing condition.
Consequently, $B_{25}\sqcup L_{25}(h)$ is a kissing configuration
throughout the interval
\begin{equation}
  \frac12\leq h\leq\frac1{\sqrt3}.
  \label{eq:d25-height-interval}
\end{equation}
In particular, the lifted layers can be moved continuously towards the ``equator''
$({\mathbb R}^{24}\times \{0\})\cap {\mathbb S}^{24}$, while the bulk remains fixed.

Let $e_{25}=(0,1)\in\mathbb R^{24}\oplus\mathbb R$.
Both poles $\pm e_{25}$ are orthogonal to $B_{25}$, and their inner
products with the lifted vectors are $\pm h$.
Thus the poles can be adjoined when $h\leq \frac12$.
On the other hand, \eqref{eq:d25-paired-lifts} requires $h\geq \frac12$,
so the height compatible with both poles is exactly $h=\frac12$.
Define the modified lifted and auxiliary blocks by
\begin{align}
  L_{25}^{\ast}
    &=L_{25}(\frac12)
      =\left\{\left(\frac{\sqrt3}{2}u,
                          \frac{\varepsilon}{2}\right)
        \;\middle|\;u\in S_1,\ \varepsilon\in\{-1,1\}\right\},
      \label{eq:d25-modified-lifted-block}\\
  F_{25}^{\ast}
    &=\{(0,-1),(0,1)\}.
      \label{eq:d25-modified-auxiliary-block}
\end{align}
It follows that
\[
  B_{25}\sqcup L_{25}^{\ast}\sqcup F_{25}^{\ast}
\]
is a kissing configuration of size
\begin{equation}
  (196560-496)+2\cdot496+2=197058.
  \label{eq:d25-improved-size}
\end{equation}
Hence $\tau_{25}\geq197058$. The resulting arrangement is schematically drawn on Figure~\ref{fig:d25}.

\begin{figure}[t]
  \centering
  \includegraphics[width=0.8\columnwidth]{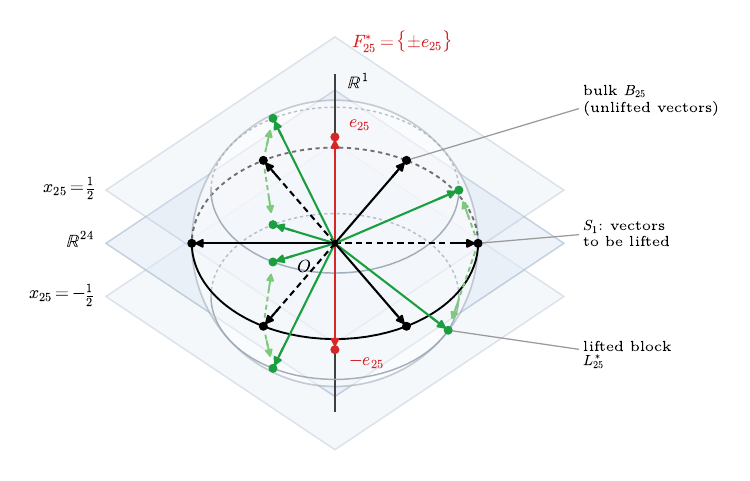}
  \caption{Three-dimensional schematic of the construction for $d=25$. The horizontal
  planes represent $\mathbb{R}^{24}$ (through the origin $O$) and the two layers
  $x_{25}=\pm\frac12$, while the vertical axis represents the additional factor
  $\mathbb{R}^{1}$; the sphere is the unit sphere, whose intersection with
  $\mathbb{R}^{24}$ is the black circle. The six black vectors pointing to the vertices
  of a hexagon inscribed in that circle represent the Leech lattice shell, with the
  three dashed vectors indicating the subset $S_{1}$ selected for lifting. Each of them
  is lifted, along the light green arrows, to the two vectors
  $\ell_{\pm}(u;\frac12)=(\frac{\sqrt3}{2}u,\pm\frac12)$ shown in green, which lie on
  the unit sphere in the layers $x_{25}=\pm\frac12$. At this height the two poles
  $\pm e_{25}$ (red) can be adjoined.}
  \label{fig:d25}
\end{figure}

\section{The case of $d=31$}
\label{sec:specifics-d31}

For $d=31$, the construction of Lemma~\ref{leech-lifting} uses $42$ disjoint
$496$-point subsets $S_i\subset\mathcal L_{24}$, with each part $T_i\subset\mathbb S^6$ an equilateral triangle. Also, $\cup_{i=1}^{42} T_i$ is isometric to the $E_7$ root
system. Thus
$|K_7|=|K'_7|=126$, and the three blocks have cardinalities
\[
 |B_{31}|=175728,\qquad |L_{31}|=62496,\qquad |F_{31}|=126.
\]
Their union is the $238350$-point configuration of~\cite{PackingStar}.

It turns out that four points can be appended to the arrangement after rotating the additional coordinates
of all lifted vectors while keeping $B_{31}$ and $F_{31}$ fixed.
For $R\in\operatorname{SO}(7)$, set
\begin{equation}
 L_{31}(R)=\bigsqcup_{i=1}^{42}
 \left\{\left(\sqrt{\frac23}\,s,\frac{Rt}{\sqrt3}\right)
       \mathrel{\Big|} s\in S_i,\ t\in T_i\right\}.
 \label{eq:d31-rotated-lifting}
\end{equation}
A common orthogonal transformation $R$ preserves all inner products
within the lifted block. It also preserves all inner products between
$B_{31}$ and the lifted block, since the bulk has zero at additional
coordinates. Consequently, the only constraints on existing vectors
that we need to satisfy are
\begin{equation}
 \frac{p^\top Rt}{\sqrt3}\leq\frac12
 \qquad(t\in K_7,\ p\in K'_7).
 \label{eq:d31-compatibility}
\end{equation}

The auxiliary configuration $K'_7$ is isometric to the $E_7$ root
system normalized to unit length. Its set $\mathcal H_7\subset\mathbb S^6$
of deepest spherical holes consists of $56$ vectors, forming $28$
antipodal pairs. 
This suggests searching for a common rotation $R$ and, for as many
holes $v\in\mathcal H_7$ as possible, a unit vector
$u_v\in\mathbb S^{23}$ such that both points
\begin{equation}
 z_v^\pm=\left(\pm\frac12u_v,\frac{\sqrt3}{2}v\right)
 \label{eq:d31-hole-pairs}
\end{equation}
can be added. Our search found a common rotation and two such unit vectors
$u_1,u_2\in\mathbb S^{23}$ for two antipodal holes $v_1=-v_2$. In the coordinates used for the
certificate, these holes are
\[
 v_1=-\left(\frac1{\sqrt3},0,0,0,0,0,\sqrt{\frac23}\right),
 \qquad v_2=-v_1.
\]
We therefore add the four points
\begin{equation}
 Z=\{z_1^+,z_1^-,z_2^+,z_2^-\},
 \qquad
 z_j^\pm=\left(\pm\frac12u_j,\frac{\sqrt3}{2}v_j\right)
 \quad(j=1,2).
 \label{eq:d31-added-points}
\end{equation}
The vectors $u_1,u_2$, the rotation matrix, and the verification data
are available in the GitHub repository accompanying this paper.
The antipodality of the two holes also gives the stronger bound
\[
 (z_1^\sigma)^\top z_2^\tau\leq\frac14-\frac34=-\frac12
 \qquad(\sigma,\tau\in\{+,-\}).
\]
The numerical bounds in Table~\ref{tab:d31-certificate} are certified
using exact integer coordinates and outward-rounded rational intervals,
including the normalization and polar-factor corrections.
The other bounds follow from the lifting construction and the
calculations above.

\begin{table}[htbp]
 \centering
 \begin{tabular}{lc}
  \hline
  Pair of blocks or points & Certified upper bound \\
  \hline
  Unchanged pairs of distinct existing vectors & $1/2$ \\
  $L_{31}(R)$, $F_{31}$ & $0.499868403188$ \\
  $B_{31}$, $Z$ & $1/2$ \\
  $L_{31}(R)$, $Z$ & $0.499383903054$ \\
  $F_{31}$, $Z$ & $1/2$ \\
  $z_j^+$, $z_j^-$, $j=1,2$ & $1/2$ \\
  $z_1^\sigma$, $z_2^\tau$, $\sigma,\tau\in\{+,-\}$ & $-1/2$ \\
  \hline
 \end{tabular}
 \caption{Inner-product bounds for the $238354$-point configuration.
 Decimal bounds are rounded upwards.}
 \label{tab:d31-certificate}
\end{table}
Consequently,
\[
 B_{31}\sqcup L_{31}(R)\sqcup F_{31}\sqcup Z
\]
is a kissing configuration of size $238350+4=238354$, proving
\begin{equation}
\tau_{31}\geq 238354.
 \label{eq:d31-new-bound}
\end{equation}

\section{The case of $d=30$}\label{sec:d30}

For $d=30$ the construction of Lemma~\ref{leech-lifting} used in~\cite{PackingStar}
takes $r=24$ pairwise disjoint $496$-point subsets
$S_1,\dots,S_{24}\subset\mathcal{L}_{24}$ satisfying~\eqref{eq:leech-subcodes}.
The configuration $K_6=T_1\sqcup\cdots\sqcup T_{24}\subset\mathbb{S}^5$ is the
normalized $E_6$ root system, partitioned into $24$ equilateral triangles, and
$K_6'$ is also isometric to the normalized $E_6$ root system with
$m(K_6',K_6)=\frac{\sqrt3}{2}$. The three blocks have cardinalities
\[
|B_{30}|=196560-24\cdot496=184656,\qquad |L_{30}|=3\cdot24\cdot496=35712,\qquad |F_{30}|=72,
\]
and their union is the $220440$-point configuration of~\cite{PackingStar}.
Unfortunately, none of the previous mechanisms applies directly. 

\subsection{The deformation}
We search over a family that combines deformations of the lifted block,
while the bulk $B_{30}$ and the auxiliary block $F_{30}$ stay fixed. For an orthogonal matrix $R\in \operatorname{O}(6)$ and a
rotation $G\in \operatorname{SO}(30)$, set
\begin{align}
L_{30}(R) &= \bigsqcup_{k=1}^{24}
  \left\{\left(\sqrt{\tfrac23}\,s,\ \tfrac{1}{\sqrt3}\,Rt\right)\ \middle|\ s\in S_{k},\ t\in T_k\right\},
  \label{eq:L30}\\
C_{30}(R,G) &= B_{30}\sqcup F_{30}\sqcup G\,L_{30}(R).
  \label{eq:C30}
\end{align}
Thus, $R$ rotates the
additional coordinates of all lifted vectors (as in Section~\ref{sec:specifics-d31}), and
$G$ rotates the whole lifted block in $\mathbb{R}^{30}$. The PackingStar
configuration is $C_{30}(I_6,I_{30})$.

For all $R\in \operatorname{O}(6)$ and $G\in \operatorname{SO}(30)$, any two
distinct vectors lying both in $B_{30}\sqcup F_{30}$ or both in
$G\,L_{30}(R)$ have inner product at most $\frac12$. Consequently,
$C_{30}(R,G)$ is a kissing configuration if and only if
\begin{equation}\label{eq:d30-constraints}
m\bigl(F_{30},\,G\,L_{30}(R)\bigr)\le\tfrac12
\quad\text{and}\quad
m\bigl(B_{30},\,G\,L_{30}(R)\bigr)\le\tfrac12 .
\end{equation}

For $R=I_6$ and $G=I_{30}$, the first condition
in~\eqref{eq:d30-constraints} holds with equality. The second condition holds
automatically when $\|G-I\|_2\le\frac12-\frac1{\sqrt6}\approx0.0918$,
by~\eqref{eq:initial-cross-slack} and~\eqref{eq:rotation-neighborhood}. The rotation we found
is larger than this, and the following lemma reduces the second condition to a
finite set of explicitly checkable pairs.

\begin{lemma}\label{lem:d30-tracked}
Suppose $\|G-I\|_2\le\frac12-\frac{\sqrt6}{12}\approx0.2959$. Then
$m\bigl(B_{30},G\,L_{30}(R)\bigr)\le\frac12$ holds if and only if
$(u,0)^\top G\ell\le\frac12$ for every $u\in\mathcal{L}_{24}\setminus\bigcup_iS_i$
and every lift $\ell=\bigl(\sqrt{\frac23}\,s,\ \frac{1}{\sqrt3}Rt\bigr)$ of a vector $s\in\bigcup_iS_i$ with $u^\top s=\frac12$.
\end{lemma}

\begin{proof}
For $\ell=\bigl(\sqrt{\frac23}\,s,\ \frac{1}{\sqrt3}Rt\bigr)$ we have
$(u,0)^\top G\ell=\sqrt{\frac23}\,u^\top s+(u,0)^\top(G-I)\ell
\le \sqrt{\frac23}\,u^\top s+\|G-I\|_2$. Distinct vectors of $\mathcal{L}_{24}$
satisfy $u^\top s\in\{-1,-\frac12,-\frac14,0,\frac14,\frac12\}$, and for
$u^\top s\le\frac14$ the right-hand side is at most
$\frac{\sqrt6}{12}+\|G-I\|_2\le\frac12$.
\end{proof}

We search for $R,G$ and $x$, where $x$ is a hole of $C_{30}(R,G)$, i.e., a vector
$x\in\mathbb{S}^{29}$ with $m\bigl(C_{30}(R,G),\{x\}\bigr)\le\frac12$.
\subsection{Search}

We minimize $m\bigl(C_{30}(R,G),\{x\}\bigr)$ over
$(R,G,x)$ subject to~\eqref{eq:d30-constraints}. Each run starts from a rotation $R$ taken from a pool
of local minimizers of $m(K_6',RK_6)$, whose smallest value is $0.7848<\frac{\sqrt3}{2}$.
Such an $R$ releases all contacts between the lifted and auxiliary blocks.
We used a sequential linear programming
method with a trust region, as in Section~\ref{sec:experimental-results}.

The final solution has a rotation $G$ with
$\|G-I\|_2\approx  0.2598$. The resulting hole $x$ satisfies
$m\bigl(C_{30}(R,G),\{x\}\bigr)\approx  0.48732$. Adding the antipode of $x$ does not violate kissing constraints.
The inner-product bounds between the blocks are listed in
Table~\ref{tab:d30}.

\begin{table}[t]
\centering
\caption{Inner-product bounds for the $220442$-point configuration
$B_{30}\sqcup F_{30}\sqcup G\,L_{30}(R)\sqcup\{x,-x\}$.
Decimal bounds are rounded upwards.}
\label{tab:d30}
\begin{tabular}{lc}
\hline
Pair of blocks or points & Upper bound\\
\hline
Distinct vectors within $B_{30}\sqcup F_{30}$ or within $G\,L_{30}(R)$ & $1/2$\\
$B_{30}$, $G\,L_{30}(R)$, pairs with $u^\top s=\frac12$ & $0.499999797$\\ 
$B_{30}$, $G\,L_{30}(R)$, pairs with $u^\top s\le\frac14$ & $0.463920$\\ 
$F_{30}$, $G\,L_{30}(R)$ & $0.499999901$\\ 
$\{x,-x\}$, $C_{30}(R,G)$ & $0.487319117$\\ 
\hline
\end{tabular}
\end{table}

Consequently,
$B_{30}\sqcup F_{30}\sqcup G\,L_{30}(R)\sqcup\{x,-x\}$
is a kissing configuration of size $220440+2=220442$, proving
\begin{equation}\label{eq:tau30}
\tau_{30}\ge 220442 .
\end{equation}

The certificate consists of the vector $x$ and the matrices $R\in \operatorname{O}(6)$ and
$G\in \operatorname{SO}(30)$. The files containing $x$ and the resulting matrices $R$
and $G$ are available in the accompanying GitHub repository.

\bibliographystyle{IEEEtran}
\bibliography{lit}
\end{document}